\documentclass[a4paper,reqno]{amsart}
\usepackage{amssymb,latexsym,amsmath,amsthm,mathtools}
\usepackage{euscript}
\usepackage{graphics,color}
\usepackage[margin=3cm]{geometry}
\usepackage{tikz}
\usetikzlibrary{calc}
\usepackage{enumitem}
\usepackage{caption}
\usepackage{hyperref}
\hypersetup{colorlinks=true,linkcolor=blue!55!black,citecolor=teal,urlcolor=teal}

\newcommand{\J}{\mathfrak{J}}
\newcommand{\R}{\mathfrak{R}}
\newcommand{\su}{\mathfrak{su}}

\newcommand{\be}{\begin{equation}}
\newcommand{\ee}{\end{equation}}
\newcommand{\dx}{\partial_x}\newcommand{\dy}{\partial_y}

\newcommand{\hF}{\,{}_2F_1}
\newcommand{\hP}[2]{\widehat P^{(#1)}_{#2}}

\newcommand{\Lone}{L_1}\newcommand{\Lthree}{L_3}
\newcommand{\Xone}{X_1}\newcommand{\Xthree}{X_3}
\newcommand{\anticomm}[2]{\{#1,#2\}}

\newcommand{\argu}[3]{\left(\begin{matrix} #1\\#2\end{matrix}\,;\, #3\right)}

\newtheorem{theorem}{Theorem}[section]

\theoremstyle{definition}

\numberwithin{equation}{section}

\title[Realization embeddings into the rank two Jacobi algebra]
  {Realization embeddings of the rank two Racah algebra into the rank two Jacobi algebra}
\author[N.~Crampé, S.~Post, L.~Vinet]{Nicolas Crampé, Sarah Post, Luc Vinet}
\address{Université de Tours, Université d’Orléans, CNRS, IDP, UMR 7013, 37200 Tours, France}
\address{Laboratoire d'Annecy-le-Vieux de Physique Th\'eorique (LAPTh), CNRS, Universit\'e Savoie Mont Blanc,
BP 110, F-74941 Annecy-le-Vieux, France}
\address{Department of Mathematics, University of Hawaii at Manoa, Honolulu, HI 96822, USA}
\address{IVADO and Centre de recherches math\'ematiques, Universit\'e de Montr\'eal, P.O. Box 6128,
Centre-ville Station, Montr\'eal (Qu\'ebec), H3C 3J7, Canada}

\begin{document}

\begin{abstract}
We construct an explicit embedding of the rank two Racah algebra $\R_2$ into the rank two Jacobi
algebra $\J_2$. Working in the differential model of $\J_2$ on the triangle, we recall that the
subalgebra structure of $\J_2$ is organized by a pentagon whose four upper edges carry rank one
Jacobi algebras and whose base carries a rank one Racah algebra. Promoting the two multiplication
operators to Racah type generators by tridiagonalization turns every edge of
the pentagon into a rank one Racah algebra and realizes $\R_2$ inside $\J_2$. We determine the
common eigenfunctions of different commuting pairs explicitly as products
of Gauss hypergeometric functions and Jacobi polynomials. Different overlap coefficients are computed involving univariate Wilson polynomials as well as, for one pair of bases, bivariate Racah polynomials of Tratnik type. We also
recall how $\J_2$ arises as a contraction of the rank two Racah algebra $\R_2$.
\end{abstract}

\maketitle

\section{Introduction}\label{sec:intro}
We shall respectively denote by $\R_2$ and $\J_2$ the rank two Racah and Jacobi algebras. A
definition of the former can be found in \cite{de2017higher} and the latter has been recently
introduced in \cite{crampe2025rank}. According to their nomenclature these algebras are associated
with the bivariate Racah and Jacobi polynomials. We shall describe how $\R_2$ embeds in $\J_2$ and
shall exploit a model of $\J_2$ in terms of second order partial differential operators in two
variables to develop the characterization of the bivariate special functions that arise in this
framework.

In the univariate realm, it is well known that the Wilson and Racah polynomials sit at the top of the
Askey scheme classifying the bispectral hypergeometric orthogonal polynomials
\cite{koekoek2010hypergeometric}. All the other families in that scheme are obtained from these as
limits or specializations. At the level of the algebras encoding the bispectrality of these families,
the corresponding picture has the lower structures arising from the rank one Racah algebra $\R_1$
through contractions or special parameter values (see \cite{koornwinder2018dualities} for instance).
This provides a top--to--bottom view.

There is, alternatively, a bottom--to--top framework that goes under the name of
\emph{tridiagonalization} and was initiated by Ismail and Koelink \cite{ismail2012spectral}. In a
nutshell, the operator defining the eigenvalue equation of a higher family is constructed as a
quadratic expression in the bispectral operators of a lower family. This was carried out explicitly in
\cite{genest2016tridiagonalization} to obtain the Wilson and Racah polynomials from the Jacobi ones.
From the algebraic standpoint this provides an embedding of the higher algebra into the lower one ---
of the Racah algebra into the Jacobi algebra \cite{genest2016tridiagonalization} --- and yields a model
of the Racah algebra in terms of second order differential operators in one variable
\cite{gao2013classification}. This possibility was explained in \cite{genest2013equitable} through the
dimensional reduction, from three variables to one, of the realization of the Racah algebra in terms
of intermediate Casimir operators in three Bargmann representations of $\su(1,1)$, and was extended in
\cite{de2017higher} to the $n$-fold tensor product so as to realize $\R_n$ in terms of second order
differential operators in $n$ variables. This last remark is germane to the present objective.

We aim here to extend this bottom--up framework to the rank two algebras. The rank two Jacobi algebra
$\J_2$ was introduced recently \cite{crampe2025rank} on the basis of the bispectral properties of the
two-variable Jacobi polynomials \cite{koornwinder1975two,dunkl2014orthogonal}. Reversing the
perspective, it was further shown \cite{crampe2025algebraic} that the representations of $\J_2$ provide
an algebraic interpretation of these polynomials that is covariant under the dihedral group, and that
$\J_2$ is the dynamical algebra of the generic superintegrable model on the $2$-sphere
\cite{crampe2025super}. In the following we provide the embedding of $\R_2$ into $\J_2$, work out
explicitly the bases of the two-variable realization whose overlaps are Wilson and Racah polynomials.

The paper is organized as follows. Section~\ref{sec:rank1} reviews the rank one tridiagonalization \cite{genest2016tridiagonalization}  and the rank one Racah algebra it produces, together with the
closed-form hypergeometric eigenfunctions and the Wilson functions as overlap coefficients. Section~\ref{sec:embed}
recalls the differential model of $\J_2$ and its pentagonal subalgebra structure, defines the Racah
type generators, and gives the explicit correspondence with $\R_2$. Section~\ref{sec:bases}
determines common eigenfunctions of commuting subalgebras in $\R_2$ and Section~\ref{ss:overlaps} provides some overlap coefficients between these eigenfunctions and connect with the bivariate polynomials of Tratnik. Concluding remarks are offered in Section~\ref{sec:conc}. The
defining relations of $\J_2$ and of $\R_2$  are collected in Appendices~\ref{APP A} and \ref{app:algR4}, respectively.

\section{Tridiagonalization in the rank one case}\label{sec:rank1}
We recall the one-variable mechanism that will be lifted twice in the sequel.

\subsection{The Jacobi algebra and the hypergeometric operator}
Let introduce the hypergeometric operator
\begin{equation}
\ell=x(1-x)\partial_x^{2}+\bigl(\alpha+1-(\alpha+\beta+2)x\bigr)\dx\,,
\label{eq:hyp1}
\end{equation}
depending on the parameters $\alpha$ and $\beta$,
whose eigenfunctions are the Jacobi polynomials. More precisely, the following eigenvalue problems hold, for $n=0,1,2,\dots$,
\begin{equation}
    \ell \hP{\alpha,\beta}{n}(x)=-n(n+\alpha+\beta+1)\hP{\alpha,\beta}{n}(x)\,,
\end{equation}
with
\begin{equation}
    \hP{\alpha,\beta}{n}(x)=
    \hF\!\argu{-n,\,n+\alpha+\beta+1}{\alpha+1}{x}\,.
\end{equation}
The polynomials $\hP{\alpha,\beta}{n}(x)$ satisfy the following orthogonality relation:
\begin{equation}
\int_0^1\hP{\alpha,\beta}{n}(x)\hP{\alpha,\beta}{m}(x)\,x^{\alpha}(1-x)^{\beta}\,dx
=\delta_{n,m}N_n^{(\alpha,\beta)},
\label{eq:ortho}
\end{equation}
with the norm
\begin{equation}
N_n^{(\alpha,\beta)}=\frac{n!\, \Gamma(\alpha+1)\,\Gamma(n+\beta+1)}{(2n+\alpha+\beta+1)\,(\alpha+1)_n\,\Gamma(n+\alpha+\beta+1)}.
\label{eq:Nnorm}
\end{equation}
For later convenience, we introduce the following transformation of the Jacobi operator $\ell$:
\begin{subequations}
   \begin{align}
\mathbf c_{12}&=-\ell\big|_{x\to x/\epsilon}+\frac{(\alpha+\beta+1)^2-1}{4}\\
&=x(x-\epsilon)\partial_x^{2}-\bigl((\alpha+1)\epsilon-(\alpha+\beta+2)x\bigr)\dx
              +\frac{(\alpha+\beta+1)^2-1}{4},
\label{eq:c12}
\end{align}
\end{subequations}
where $\epsilon$ is a free parameter.
The eigenvalue problem for this operator reads
\begin{align}\label{eq:c12P}
  \mathbf c_{12} \hP{\alpha,\beta}{n}\left(\frac{x}{\epsilon}\right)=
  \frac{(2n+\alpha+\beta+1)^2-1}{4}
 \hP{\alpha,\beta}{n}\left(\frac{x}{\epsilon}\right) \,.
\end{align}

Let $X$ denote the operator multiplication by $x$. The operators $X$ and
$\mathbf c_{12}$ close into a quadratic algebra.
Explicitly, these operators satisfy the following commutation relations:
\begin{subequations}\label{eq:J1}
    \begin{align}
    [X,[X,\mathbf{c_{12}}]]&=2X^2-2\epsilon X\,,\\
     [\mathbf c_{12},[\mathbf c_{12},X]]&=2\{X,\mathbf c_{12}\}-2\epsilon \mathbf c_{12}-2\epsilon(\mathbf{c_1}-\mathbf{c_2})\,,
\end{align}
\end{subequations}
with
\begin{align}
    \mathbf{c_1}=\frac{\alpha^2-1}{4}\,,\quad \mathbf{c_2}=\frac{\beta^2-1}{4}\,.
\end{align}

This algebra encodes the bispectrality of the Jacobi polynomials and is called the bispectral algebra associated to the Jacobi polynomials or \emph{Jacobi algebra} $\J_1$ (of rank $1$).
Indeed, the overlap coefficients between the eigenbasis of $\mathbf{c_{12}}$ (the Jacobi
polynomials) and that of $X$ (Kronecker functions) are the Jacobi polynomials themselves.

\subsection{Tridiagonalization and the rank one Racah algebra\label{ssec:trid1}}
To any bispectral algebra, we can associate the so-called algebraic Heun operator \cite{grunbaum2017tridiagonalization,BaseilhacHeunAW,CrampeHeunLie,BergeronHeunRacah2020}:
\begin{equation}
M=\tau_1\,[X,\mathbf c_{12}]+\tau_2\,\{\mathbf c_{12}, X\}+\tau_3\,X+\tau_4\, \mathbf c_{12} +\tau_0\,,
\label{eq:Mgivz}
\end{equation}
for some parameter $\tau_i\ (i=0,1,2,3,4)$. It is the most general bilinear expression, in the bispectral operators, here $\mathbf{c_{12}}$ and $X$ for the Jacobi algebra. $M$ is tridiagonal in the eigenbasis of $X$ but also in the one of $\mathbf c_{12}$: it is why the procedure to construct $M$ is called tridiagonalisation. 
Following  \cite{genest2016tridiagonalization},
we know that, for a particular choice of the parameters $\tau_i$, the algebraic Heun operators have special properties. Indeed, let us specialize \eqref{eq:Mgivz} to:
\begin{align}\label{eq:Heun1}
    \mathbf{c_{23}}=\frac{\delta-\alpha}{2\epsilon}[X,\mathbf{c_{12}}]-\frac{1}{2\epsilon}\{X,\mathbf{c_{12}}\}-\frac{(\delta-\alpha)^2-\gamma^2}{4\epsilon}X+\frac{\alpha^2+(\delta-\alpha)^2-2}{4}\,,
\end{align}
for $\gamma$ and $\delta$ some free parameters.
This operator reads as follows in terms of $x$ and $\partial_x$:
\begin{align}
    \mathbf{c_{23}}&=x^2\left(1-\frac{x}{\epsilon}\right)\partial_x^2+x\left( \delta + 2-( \beta+\delta + 3 )\frac{x}{\epsilon}\right)\partial_x+\frac{\gamma^2-(\beta+\delta+2 )^2}{4}\, \frac{x}{\epsilon}+\frac{( \delta+1 )^2-1}{4}\,.
\end{align}
The couple of operators $(\mathbf{c_{12}},\mathbf{c_{23}})$ satisfy the following relations
\begin{subequations}\label{eq:R1}
    \begin{align}[\mathbf{c_{12}},[\mathbf{c_{12}},\mathbf{c_{23}}]]&=2\{\mathbf{c_{12}},\mathbf{c_{23}}\}+2\mathbf{c_{12}}^2-2(\mathbf{c_1}+\mathbf{c_2}+\mathbf{c_3}+\mathbf{c_{123}})\mathbf{c_{12}}-2(\mathbf{c_1}-\mathbf{c_{2}})(\mathbf{c_3}-\mathbf{c_{123}})\,,\\
[\mathbf{c_{23}},[\mathbf{c_{23}},\mathbf{c_{12}}]]&=2\{\mathbf{c_{12}},\mathbf{c_{23}}\}+2\mathbf{c_{23}}^2-2(\mathbf{c_1}+\mathbf{c_2}+\mathbf{c_3}+\mathbf{c_{123}})\mathbf{c_{23}}-2(\mathbf{c_1}-\mathbf{c_{123}})(\mathbf{c_3}-\mathbf{c_{2}})\,,
\end{align}
\end{subequations}
with the four parameters given by
\begin{align}
    &\mathbf{c_1}=\frac{\alpha^2-1}{4}\,,\qquad \mathbf{c_2}=\frac{\beta^2-1}{4}\,,\qquad \mathbf{c_3}=\frac{\gamma^2-1}{4}\,,\qquad \mathbf{c_{123}}=\frac{(\delta-\alpha)^2-1}{4}\,.
\end{align}

This algebra is called Racah algebra (of rank 1) $\R_1$ and is the bispectral algebra associated to the Racah polynomials \cite{GranovskiiRacah}.

There exists another algebraic Heun operator which allows us to obtain a realization of the Racah algebra of rank 1.
More precisely, let define the operator
\begin{align}
  \mathbf{c'_{23}}=  -\mathbf{c_{23}}-\mathbf{c_{12}}+\mathbf{c_1}+\mathbf{c_2}+\mathbf{c_3}+\mathbf{c_{123}}\,,
\end{align}
which reads as
\begin{align}\label{eq:Heun2}
  \mathbf{c'_{23}}=-\frac{\delta-\alpha}{2\epsilon}[X,\mathbf{c_{12}}]+\frac{1}{2\epsilon}\{X-\epsilon,\mathbf{c_{12}}\}+\frac{(\delta-\alpha)^2-\gamma^2}{4\epsilon}(X-\epsilon)+\frac{\beta^2+(\delta-\alpha)^2-2}{4}\,.
\end{align}
The couple $( \mathbf{c_{12}}, \mathbf{c'_{23}})$ also satisfies the defining relations of the Racah algebra by exchanging $\mathbf{c_{123}}\leftrightarrow \mathbf{c_{3}}$ in \eqref{eq:R1}.
 This recaps the fact observed previously \cite{grunbaum2017tridiagonalization} that the Racah algebra can be realized in two different manners in terms of the generators of the Jacobi algebra through the construction of a particular Heun operator.

\subsection{Eigenfunctions of \texorpdfstring{$\mathbf{c_{23}}$}{c23} and overlap coefficients}

It has been demonstrated in  \cite{genest2016tridiagonalization} that the functions
\begin{equation}
\psi^{(\beta,\gamma,\delta)}_m(x)=x^{m-\frac12(\delta+1)}\,
\hF\!\argu{m+\frac12( \beta+\gamma+ 1),\ m+\frac12(\beta-\gamma +1)}{\beta+1}{1-x}\,,
\label{eq:givz36}
\end{equation}
are eigenfunctions of $\mathbf{c_{23}}$:
\begin{align}\label{eq:c23psi}
    \mathbf{c_{23}}\psi^{(\beta,\gamma,\delta)}_m\left(\frac{x}{\epsilon}\right)=\frac{4m^2-1}{4}\,\psi^{(\beta,\gamma,\delta)}_m\left(\frac{x}{\epsilon}\right)\,,
\end{align}
It is important to notice that the previous eigenfunctions are square-integrable with the Jacobi weight associated to the Jacobi polynomials $\hP{\alpha,\beta}{n}$.

As the operators  $\mathbf{c_{12}}$ and $\mathbf{c_{23}}$ satisfy the defining relations of the  Racah algebra, we expect that the overlap coefficients between their respective eigenfunctions are proportional to the Racah or Wilson polynomials.
In \cite{genest2016tridiagonalization}, it has been shown that
the expansion of $\psi_m$ in the basis
$\hP{\alpha,\beta}{n}$ has Wilson polynomials as coefficients. In the notation used in this paper, this expansion reads:
\begin{subequations}
    \begin{align}\label{eq:psiP}
    \psi^{(\beta,\gamma,\delta)}_m(x)= \sum_{n=0}^\infty  W^{(\alpha,\beta,\gamma,\delta)}_n\left(m\right) \,\hP{\alpha,\beta}{n}(x)\,,
\end{align}
where the Wilson polynomials are defined by\footnote{There is a change of variable $m= ix$ in comparison to the Wilson polynomials.}
\begin{align}
   W^{(\alpha,\beta,\gamma,\delta)}_n(m) &=\ \mu^{(\alpha,\beta,\gamma,\delta)}_m \sigma^{(\alpha,\beta)}_n\\
   &\times {}_4F_3\!\argu{-n,n+\alpha+\beta+1,\frac12(2\alpha-\delta+1)+m,\frac12(2\alpha-\delta+1)-m}{\alpha+1,\frac12(2\alpha+\beta-\gamma-\delta+2),\frac12(2\alpha+\beta+\gamma-\delta+2)}{1}\,,\nonumber
\end{align}
\end{subequations}
and the normalizations are given by
\begin{align}
    \mu^{(\alpha,\beta,\gamma,\delta)}_m&=\frac{1}{N_0^{(\alpha,\beta)}}\int_0^1 \psi^{(\beta,\gamma,\delta)}_m(x)x^\alpha(1-x)^\beta dx\,,\\
    \sigma^{(\alpha,\beta)}_n&=\frac{(2n+\alpha+\beta+1)(\alpha+1,\alpha+\beta+1)_n}{(\alpha+\beta+1)n!(\beta+1)_n}\,.
\end{align} 
Using the orthogonality relation of the Jacobi polynomials in relation \eqref{eq:psiP}, the integral expression of the Wilson polynomials from \cite{Koor,genest2016tridiagonalization} is re-derived:
\begin{align}\label{eq:KoorF}
    W^{(\alpha,\beta,\gamma,\delta)}_n(m)=\frac{1}{ N_n^{(\alpha,\beta)}}\int_0^1 \psi^{(\beta,\gamma,\delta)}_m(x) \hP{\alpha,\beta}{n}(x) x^\alpha(1-x)^\beta dx\,.
\end{align}
\section{The embedding of \texorpdfstring{$\R_2$}{R2} into \texorpdfstring{$\J_2$}{J2}}\label{sec:embed}

The results of the previous section can be generalized to the case with two variables. We show here that there exists an embedding of the rank 2 Racah algebra  $\R_2$ into the rank 2 Jacobi algebra $\J_2$.

\subsection{Differential realization of \texorpdfstring{$\J_2$}{J2}}
Let $a,\ b$ and $c$ be generic real parameters and $x,y$ two variables.
The following five operators
\begin{subequations}
    \begin{align}
\Xone&=x,\qquad \Xthree=1-x-y,\label{eq:X}\\
\Lone&=y(1-x-y)\,\partial_{y}^2+\bigl((b+1)(1-x)-(b+c+2)y\bigr)\dy,\label{eq:L1}\\
\Lthree&=xy\bigl(\partial_{x}^2+\partial_{y}^2-2\partial_{x}\partial_y\bigr)+\bigl((a+1)y-(b+1)x\bigr)(\dx-\dy),
\label{eq:L3}\\
L&=x(1-x)\partial_{x}^2+y(1-y)\partial_{y}^2-2xy\,\partial_{x}\partial_y\nonumber\\
&\hspace{1.2cm}+(a+1-(a+b+c+3)x)\dx+(b+1-(a+b+c+3)y)\dy\,.\label{eq:L}
\end{align}
\end{subequations}
satisfy the rank two Jacobi algebra $\J_2$ \cite{crampe2025rank} whose the
defining relations are recalled in Appendix~\ref{APP A}.

Let us introduce affine transformations of previous generators
\begin{subequations}
    \begin{align}
C_{23}&=-\Lone+\tfrac{(b+c+1)^2-1}{4}\,,\\
C_{34}&=-\Lthree+\tfrac{(a+b+1)^2-1}{4}\,,\\
C_{234}&=-L+\tfrac{(a+b+c+2)^2-1}{4}\,.\label{eq:dictA}
\end{align}
\end{subequations}
From the relations provided in Appendix~\ref{APP A}, we can show that different pairs of generators of $\J_2$ yield a Jacobi algebra of rank 1.
We list below these different couples with the associated parameters $(\mathbf{c_1},\mathbf{c_2},\epsilon)$ of the defining relations \eqref{eq:J1}:
\begin{itemize}
    \item $(X_1,C_{34})$ with $(\mathbf{c_1},\mathbf{c_2},\epsilon)=(\frac{a^2-1}{4},\frac{b^2-1}{4},1-X_3)$;
    \item $(X_3,C_{23})$ with $(\mathbf{c_1},\mathbf{c_2},\epsilon)=(\frac{c^2-1}{4},\frac{b^2-1}{4},1-X_1)$;
    \item $(X_1,C_{234})$ with $(\mathbf{c_1},\mathbf{c_2},\epsilon)=(\frac{a^2-1}{4},C_{23},1)$;
    \item $(X_3,C_{234})$ with $(\mathbf{c_1},\mathbf{c_2},\epsilon)=(\frac{c^2-1}{4},C_{34},1)$.
\end{itemize}
At first glance, it can be surprising that generators appear in the set of parameters of the subalgebras above. However, each time, they commute with both generators of the subalgebra and we hence actually have central extensions of the rank 1 Jacobi algebra.

The pair of generators $(C_{23},C_{34})$ forms a rank 1 Racah algebra. Indeed, we can check that, setting $(\mathbf{c_{12}},\mathbf{c_{23}})=(C_{23},C_{34})$, the defining relations \eqref{eq:R1} of $\J_1$ are satisfied with the parameters
\begin{align}
    \mathbf{c_1}=\frac{c^2-1}{4}\,,\quad\mathbf{c_2}=\frac{b^2-1}{4}\,,\quad \mathbf{c_3}=\frac{a^2-1}{4}\,,\quad \mathbf{c_{123}}=C_{234}\,.
\end{align}
As previously, a generator of $\J_2$ appears in the parameters of the defining relations of this subalgebra but $C_{234}$ commutes with $C_{23}$ and $C_{34}$ and we thus have a central extension of the rank 1 Racah algebra.
It is useful to represent these five subalgebras in a pentagon as displayed in Figure~\ref{fig:pentJ}.
\begin{figure}[htbp]
\centering
\begin{tikzpicture}[scale=1.3]
    \def\R{2cm}\def\Rt{2.15cm}
    \foreach \i in {1,...,5} {\coordinate (P\i) at ({90 + (\i-1)*72}:\R);}
    \draw (P3)--(P2)--(P1)--(P5)--(P4);
    \draw[dashed] (P3)--(P4);
    \node at ({90 + (+0.13)*72}:\Rt) {$X_3$};   \node at ({90 + (-0.13)*72}:\Rt) {$X_1$};
    \node at ({90 + (-0.13+1)*72}:\Rt) {$C_{23}$};  \node at ({90 + (+0.13+1)*72}:\Rt) {$X_1$};
    \node at ({90 + (-0.13-1)*72}:\Rt) {$X_3$};  \node at ({90 + (+0.13-1)*72}:\Rt) {$C_{34}$};
    \node at ({90 + (-0.13-2)*72}:\Rt) {$C_{34}$};  \node at ({90 + (+0.13-2)*72}:\Rt) {$C_{234}$};
    \node at ({90 + (-0.13+2)*72}:\Rt) {$C_{234}$};    \node at ({90 + (+0.13+2)*72}:\Rt) {$C_{23}$};
    \foreach \i in {1,...,5} {\draw[fill] (P\i) circle (0.07);}
\end{tikzpicture}
\caption{Two generators at a vertex commute and two generators at the opposite extremities of the same edge generate one of the five subalgebras.
The four solid edges carry
rank one Jacobi algebras and the dashed base carries the rank one Racah algebra.\label{fig:pentJ}}
\end{figure}
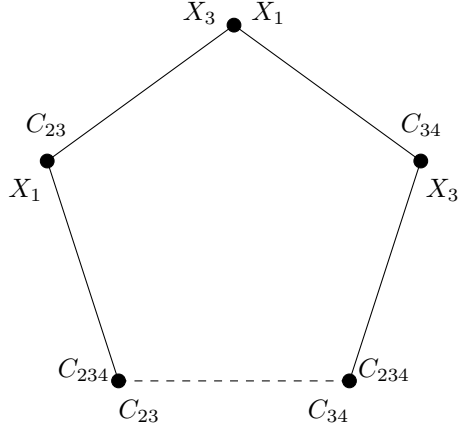

\subsection{Realization of the rank 2 Racah algebra.}

As in Subsection~\ref{ssec:trid1},
we want to promote the two multiplication operators $\Xone$ and $\Xthree$ to Racah type generators by using special Heun operators.
As mentioned previously, $(X_1,C_{234})$ forms a rank 1 Jacobi algebra. From relation \eqref{eq:Heun1} with $\alpha=a$, $\gamma=d-b-c$ and $\delta=2a+d+2$ ($d$ a new free parameter), we obtain the particular Heun operator
\begin{align}\label{eq:defC123}
   C_{123}=&\frac{a+d+2}{2}[X_1,C_{234}]-\frac{1}{2}\{X_1,C_{234}\}\nonumber\\
   -&\frac{(a+d+2)^2-(d-b-c)^2}{4}X_1+\frac{a^2+(a+d+2)^2-2}{4}\,.
\end{align}
We know, from the results of the previous section, that the couple $(C_{234},C_{123})$ satisfies the defining relations \eqref{eq:R1} of the rank 1 Racah algebra with the parameters:
\begin{align}
    \mathbf{c_1}=\frac{a^2-1}{4}\,,\quad \mathbf{c_2}=C_{23}\,,\quad \mathbf{c_3}=\frac{(d-b-c)^2-1}{4}\,,\quad \mathbf{c_{123}}=\frac{(a+d+2)^2-1}{4}\,.
\end{align}
The parameters $\gamma$ and $\delta$  in the definition \eqref{eq:defC123} of the Heun operator $C_{123}$ have been chosen carefully such that the supplementary relations of the rank 2 Racah algebra are satisfied as explained below.

Similarly, starting from the rank 1 Jacobi algebra generated by  $(X_3,C_{234})$, we can define a Heun operator from relation \eqref{eq:Heun2} with $(\beta^2-1)/4=C_{34}$, $\gamma=d-b-c$ and $\delta-\alpha=a+d+2$ which reads as follows
\begin{align}
    C_{12}=&-\frac{a+d+2}{2}[X_3,C_{234}]+\frac{1}{2}\{X_3-1,C_{234}\}+\frac{(a+d+2)^2-(d-b-c)^2}{4}(X_3-1)\nonumber\\
    +&C_{34}+\frac{(a+d+2)^2-1}{4}\,,
\end{align}
such that the couple $(C_{234},C_{12})$ forms a rank 1 Racah algebra with the parameters
\begin{align}
    \mathbf{c_1}=\frac{c^2-1}{4}\,,\quad \mathbf{c_2}=C_{34}\,,\quad \mathbf{c_3}=\frac{(a+d+2)^2-1}{4}\,,\quad \mathbf{c_{123}}=\frac{(d-b-c)^2-1}{4}\,.
\end{align}

We can also show that the couples $(C_{12},C_{23})$ and $(C_{123},C_{34})$ satisfy the defining relations \eqref{eq:R1} of the rank 1 Racah algebra with the parameters given respectively by:
\begin{align}
   & \mathbf{c_1}=\frac{(d-b-c)^2-1}{4}\,,\quad
     \mathbf{c_2}=\frac{c^2-1}{4}\,,\quad
      \mathbf{c_3}=\frac{b^2-1}{4}\,,\quad
       \mathbf{c_{123}}=C_{123}\,,\\
    &   \mathbf{c_1}=C_{12}\,,\quad
     \mathbf{c_2}=\frac{b^2-1}{4}\,,\quad
      \mathbf{c_3}=\frac{a^2-1}{4}\,,\quad
       \mathbf{c_{123}}=\frac{(a+d+2)^2-1}{4}\,.
\end{align}
A visual way to understand these five rank 1 Racah algebras consists in showing them in a pentagon as shown in Figure~\ref{fig:pentR}. In this figure, we have replaced the generators $X_1$ and $X_3$ of Figure~\ref{fig:pentJ} by the operators
$C_{123}$ and $C_{12}$ and the five edges correspond to a rank 1 Racah algebra.
\begin{figure}[htbp]
\centering
\begin{tikzpicture}[scale=1.3]
    \def\R{2cm}\def\Rt{2.15cm}
    \foreach \i in {1,...,5} {\coordinate (P\i) at ({90 + (\i-1)*72}:\R);}
    \draw[dashed] (P3)--(P2)--(P1)--(P5)--(P4);
    \draw[dashed] (P3)--(P4);
    \node at ({90 + (+0.13)*72}:\Rt) {$C_{12}$};   \node at ({90 + (-0.13)*72}:\Rt) {$C_{123}$};
    \node at ({90 + (-0.13+1)*72}:\Rt) {$C_{23}$};  \node at ({90 + (+0.13+1)*72}:\Rt) {$C_{123}$};
    \node at ({90 + (-0.13-1)*72}:\Rt) {$C_{12}$};  \node at ({90 + (+0.13-1)*72}:\Rt) {$C_{34}$};
    \node at ({90 + (-0.13-2)*72}:\Rt) {$C_{34}$};  \node at ({90 + (+0.13-2)*72}:\Rt) {$C_{234}$};
    \node at ({90 + (-0.13+2)*72}:\Rt) {$C_{234}$};    \node at ({90 + (+0.13+2)*72}:\Rt) {$C_{23}$};
    \node at ({90 + 0*72}:1.60cm) {$\Phi$};
    \node at ({90 + 1*72}:1.60cm) {$\Psi$};
    \node at ({90 + 2*72}:1.60cm) {$J$};
    \node at ({90 + 3*72}:1.60cm) {$\widetilde J$};
    \node at ({90 + 4*72}:1.60cm) {$\widetilde\Psi$};
    \foreach \i in {1,...,5} {\draw[fill] (P\i) circle (0.07);}
\end{tikzpicture}
\caption{Both generators on the same vertex commute and at the extremities of the same dashed edge form a rank 1 Racah algebra. The joint eigenfunctions of the two commuting generators sitting at an edge is indicated inside the pentagon: $J$, $\Psi$, $\Phi$, $\widetilde\Psi$ and $\widetilde J$. These bases are defined in Section~\ref{sec:bases}. The overlap between the two bases at the endpoints of an edge is a Racah polynomial (see Section~\ref{ss:overlaps}).}\label{fig:pentR}
\end{figure}
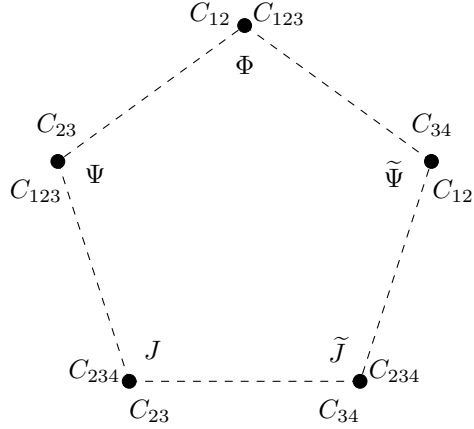

With these relations of the rank 1 Racah algebras, we can show that the five generators defined above $C_{12},\,C_{23},\,C_{34},\,C_{123},\,C_{234}$
satisfy all the relations recalled in Appendix~\ref{app:algR4}, with the parameters given by
\begin{equation}
C_1=\tfrac{(d-b-c)^2-1}{4},\ \ C_2=\tfrac{c^2-1}{4},\ \ C_3=\tfrac{b^2-1}{4},\ \
C_4=\tfrac{a^2-1}{4},\ \ C_{1234}=\tfrac{(a+d+2)^2-1}{4}.
\label{eq:dictCent}
\end{equation}
This proves that these five generators verify the defining relations of the rank 2 Racah algebra \cite{de2017higher,crampe2021racah,crampe2023representations,post2024racah}. From this construction, we obtain a differential realization of the rank 2 Racah algebra and an embedding of this algebra inside the rank 2 Jacobi algebra.
Let us remark that the pentagon shown in Figure~\ref{fig:pentR} is only one part of a larger picture, an half-icosidodecahedron, allowing to visualize the defining relations of the rank 2 Racah algebra \cite{crampe2023representations}.

\section{Explicit bases\label{sec:bases}}

The common eigenfunctions for the different couples of commuting operators appearing in Figure~\ref{fig:pentR} are computed as well as different overlap coefficients between them. This relation generalize the Koornwinder relation \eqref{eq:KoorF} to the case of the bivariate Jacobi polynomials.

\subsection{Eigenfunctions of \texorpdfstring{$(C_{23},C_{234})$}{(C23,C234)}}
 The two-variable Jacobi polynomials are defined by
\begin{equation}
J_{n,k}(x,y)=\hP{a,\,b+c+2k+1}{n-k}(x)\,(1-x)^{k}\,\hP{b,c}{k}\!\Bigl(\tfrac{y}{1-x}\Bigr)\,.
\label{eq:Jnk}
\end{equation}
It has been demonstrated in \cite{koornwinder1975two} that these polynomials are common eigenfunctions of $(C_{23},C_{234})$, for $0\le k\le n$:
\begin{align}
C_{234}\,J_{n,k}&=\frac{(2n+a+b+c+ 2)^2-1}{4}J_{n,k}\,,\label{eq:C234J}\\
C_{23}\,J_{n,k}&=\frac{(2k+b+c+1)^2-1}{4}J_{n,k}\,.\label{eq:C23J}
\end{align}
Let us recall the proof of these relations which is based on the change of variables
\begin{align}\label{eq:ux}
    u=x\,,\qquad v=\frac{y}{1-x}\,.
\end{align}
With these changes, the operators $C_{23}$ and $C_{234}$ become
\begin{align}
    C_{23}=&v(v-1)\partial_v^2-(b+1-(b+c+2)v)\partial_v+\frac{(b+c+1)^2-1}{4}\,,\label{eq:expC23}\\
    C_{234}=&u(u-1)\partial_u^2-(a+1-(a+b+c+3)u)\partial_u+\frac{(a+b+c+2)^2-1}{4}\nonumber\\
    +&\frac{1}{1-u}\left(v(1-v)\partial_v^2-(b+1-(b+c+2)v)\partial_v\right)\,.
\end{align}
Then, relations \eqref{eq:C234J} and \eqref{eq:C23J} can be proven using the eigenvalue problem \eqref{eq:c12P} for the univariate case.

Using the factorization \eqref{eq:Jnk}, the change of  variable \eqref{eq:ux} with $dxdy=(1-u)dudv$ and the orthogonality relation \eqref{eq:ortho} for the univariate Jacobi polynomials, we show that the polynomials \eqref{eq:Jnk} are orthogonal on the triangle $\Delta=\{(x,y)\,|\, 0\le x,\, 0\le y,\, x+y\le 1\}$ and for $a,b,c>-1$,
\begin{equation}
\int_\Delta J_{n,k}(x,y)\,J_{n',k'}(x,y)\,x^{a}y^{b}(1-x-y)^{c}\,dx\,dy=h_{n,k}\,\delta_{nn'}\delta_{kk'},\qquad
h_{n,k}=N_k^{(b,c)}\,N_{n-k}^{(a,\,b+c+2k+1)}.
\label{eq:hnk}
\end{equation}

\subsection{Eigenfunctions of \texorpdfstring{$(C_{23},C_{123})$}{(C23,C123)}\label{ss:pair1}}
The common eigenfunctions of the commuting operators $C_{23}$ and $C_{123}$ are given in the following theorem.
\begin{theorem}\label{thm:pair1}
 The function
\begin{equation}
\Psi_{m,k}(x,y)=
 \psi^{(2k+b+c+1,d-b-c,2a+d+2)}_{m}(x)\,
 (1-x)^k\,\hP{b,c}{k}\!\Bigl(\tfrac{y}{1-x}\Bigr)\,,
\label{eq:Psi1}
\end{equation}
where $\psi$ is defined by \eqref{eq:givz36},
satisfies
\begin{align}
\label{eq:C23Psi}  &  C_{23}\Psi_{m,k}=\frac{(2k+b+c+1)^2-1}{4}\Psi_{m,k}\,,\\
  &  C_{123}\Psi_{m,k}=\frac{4m^2-1}{4}\Psi_{m,k}\,.
\end{align}
\end{theorem}
\begin{proof}
Relation \eqref{eq:C23Psi} is proven using expression \eqref{eq:expC23} of $C_{23}$ and the eigenvalue problem \eqref{eq:c12P} for the univariate case.

For the second relation, let us perform the change of variable \eqref{eq:ux} in the operator $C_{123}$ to get
\begin{align}
 C_{123}=&   u^2(1-u)\partial_u^2+ (2a+d+4-(2a+b+c+d+6)u)u\partial_u
    -(a + b + c+ \frac{5}{2})(a + d + \frac{5}{2}) u\nonumber\\
    +&\frac{(2a + d +  3)^2-1}{4}-\frac{u}{1-u}(v(v-1)\partial_v^2-(b+1-(b+c+2)v)\partial_v)\,.
\end{align}
The above operator factorizes and the part with respect to $v$ is diagonalized by $\Psi_{m,k}$ using  \eqref{eq:expC23}. Then, using relation \eqref{eq:c23psi}, the second relation of the theorem is proven.
\end{proof}

\subsection{Eigenfunctions of \texorpdfstring{$(C_{34},C_{234})$}{(C34,C234)}}

Another two-variable Jacobi polynomials, obtained using the symmetry of the pentagon \cite{crampe2025algebraic},  can be defined:
\begin{equation}
\widetilde{J}_{n,k}(x,y)=\hP{a,b}{k}\!\Bigl(\tfrac{x}{x+y}\Bigr)\,(x+y)^{k}\, \hP{a+b+2k+1,c}{n-k}(x+y)\,\,.
\label{eq:Jnkt}
\end{equation}
These polynomials are common eigenfunctions of $(C_{34},C_{234})$, for $0\le k\le n$:
\begin{align}
C_{34}\,\widetilde{J}_{n,k}&=\frac{(2k+a+b+ 1)^2-1}{4}\widetilde{J}_{n,k}\,,\label{eq:C34Jt}\\
C_{234}\,\widetilde{J}_{n,k}&=\frac{(2n+a+b+c+2)^2-1}{4}\widetilde{J}_{n,k}\,.\label{eq:C234Jt}
\end{align}
The proof of these relations is based on the change of variable:
\begin{align}\label{eq:stxy}
    s=x+y\,,\qquad t=\frac{x}{x+y}\,.
\end{align}
With these changes, the operator $C_{34}$  becomes
\begin{align}
C_{34}&=t(t-1)\partial_t^2-(a+1-(a+b+2)t)\partial_t+\frac{(a + b + 1)^2-1}{4}\,,\label{eq:expC34}
\end{align}
and we obtain a new expression for $C_{234}$:
\begin{align}
    C_{234}=&s(s-1)\partial_s^2-(a+b+2-(a+b+c+3)s)\partial_s+\frac{1}{s}C_{34}\\
    -&\frac{(a + b + 1)^2-1}{4s}+\frac{(a+b+c+2)^2-1}{4}\,.\nonumber
\end{align}
Then, relations \eqref{eq:C34Jt} and \eqref{eq:C234Jt} can be proven using the eigenvalue problem \eqref{eq:c12P} for the univariate case.

Using the factorization \eqref{eq:Jnkt}, the change of  variable \eqref{eq:stxy} with $dxdy=sdsdt$ and the orthogonality relation \eqref{eq:ortho} for the univariate Jacobi polynomials, we show that the polynomials \eqref{eq:Jnkt} are orthogonal on the triangle $\Delta=\{(x,y)\,|\, 0\le x,\, 0\le y,\, x+y\le 1\}$ and for $a,b,c>-1$,
\begin{equation}
\int_\Delta\widetilde{J}_{n,k}(x,y)\,\widetilde{J}_{n',k'}(x,y)\,x^{a}y^{b}(1-x-y)^{c}\,dx\,dy=h_{n,k}\,\delta_{nn'}\delta_{kk'},\qquad
h_{n,k}=N_k^{(a,b)}\,N_{n-k}^{(a+b+2k+1,\,c)}.
\label{eq:thnk}
\end{equation}

\subsection{Eigenfunctions of \texorpdfstring{$(C_{12},C_{34})$}{(C12,C34)}\label{ss:pair5}}
The common eigenfunctions of the commuting operators $C_{12}$ and $C_{34}$ are given in the following theorem.
\begin{theorem}\label{thm:pair5}
 The function
\begin{equation}
\widetilde{\Psi}_{p,k}(x,y)=
\hP{a,b}{k}\!\Bigl(\tfrac{x}{x+y}\Bigr)\,  \psi^{(c,d-b-c,2a+b+d+3)}_{p}(x+y)\,,
\label{eq:Psi5}
\end{equation}
where $\psi$ is defined by \eqref{eq:givz36},
satisfies
\begin{align}
\label{eq:C12Psit}  &  C_{12}\widetilde{\Psi}_{p,k}=\frac{4p^2-1}{4}\widetilde{\Psi}_{p,k}\,,\\
\label{eq:C34Psit}   &  C_{34}\widetilde{\Psi}_{p,k}=\frac{(2k+a+b+1)^2-1}{4}\widetilde{\Psi}_{p,k}\,.
\end{align}
\end{theorem}
\begin{proof}
Relation \eqref{eq:C34Psit} is proven using expression \eqref{eq:expC34} of $C_{34}$ and the eigenvalue problem \eqref{eq:c12P} for the univariate case.

With the change of variable \eqref{eq:stxy}, the operator $C_{12}$ simplifies to
\begin{align}
\label{eq:expC12st} C_{12}=&s^2(1-s)\partial_s^2  +(2a+b+d+5-(2a+b+c+d+6)s) s\partial_s\\
 -& \frac{(2a + 2d + 5)(2a + 2b + 2c + 5)}{4} s+\frac{(2a+b+d+4)^2-1}{4}\,.\nonumber
\end{align}
The above operator depends only on $s$ and, using relation \eqref{eq:c23psi}, the second relation of the theorem is proven.
\end{proof}

\subsection{Eigenfunctions of \texorpdfstring{$(C_{12},C_{123})$}{(C12,C123)}}\label{ss:pair2}

\begin{theorem}\label{thm:pair2}
The functions $\Phi_{m,p}$ defined by
\begin{equation}
\Phi_{m,p}(x,y)=\psi_m^{(b,2p,2a+d+2)}\left(\frac{x}{x+y}\right)
\psi_p^{(c,d-b-c,2a+b+d+3)}(x+y)\,,
\label{eq:Phi}
\end{equation}
satisfy the following eigenvalue problems:
\begin{align}
    C_{12}\Phi_{m,p}(x,y)=\frac{4p^2-1}{4}\Phi_{m,p}(x,y)\,,\\
      C_{123}\Phi_{m,p}(x,y)=\frac{4m^2-1}{4}\Phi_{m,p}(x,y)\,.
\end{align}
\end{theorem}
\begin{proof}
With the change of variable \eqref{eq:stxy}, the operator $C_{12}$ is given by \eqref{eq:expC12st} and $C_{123}$ simplifies to
\begin{align}
 C_{123}=& t^2(1-t)\partial_t^2
 +(2a+d+4-(2a+b+d+5)t) t\partial_t-\frac{(2a+b+d + 4  )^2-1}{4}t \\
 +& tC_{12}+\frac{(2a + d + 3)^2-1}{4}\,.\nonumber
\end{align}
Using relation \eqref{eq:c23psi}, the theorem is proven.
\end{proof}

\section{Overlap coefficients and Tratnik polynomials }\label{ss:overlaps}

In previous section, common eigenfunctions to different pairs of commuting generators of $\R_2$ appearing at each vertex of the pentagon in Figure~\ref{fig:pentR} have been computed explicitly. We recall that each dashed of this pentagon of carries a rank one Racah algebra $\R_1$ whose two generators are the two
generators sitting at the opposite ends of that edge and the generator common to the two
edges is central for this pair. This implies, by the representation theory of $\R_1$ recalled in
Section~\ref{sec:rank1}, that the overlap coefficients between the two
eigenbases attached to the endpoints of an edge is therefore always a Racah (Wilson)
polynomial with the eigenvalue of the central generator entering only as a spectator
parameter. We establish these coefficients for two of the five edges for which the shared generator
acts, in suitable coordinates, as a \emph{one-variable} operator. The two bases then share a common
factor and the weight of the orthogonality relation factorizes.
We recall that 
the overlap coefficients between $J$ and $\widetilde J$ corresponds to the Dunkl formula proven in \cite{dunkl1984orthogonal}.

\medskip

\subsection{Edge \texorpdfstring{$\widetilde\Psi$--$\widetilde J$}{PsiJ} (shared generator \texorpdfstring{$C_{34}$}{C34})}\label{sss:PsitJt}

The overlap coefficients between $\widetilde{\Psi}_{m,k}$ and $\widetilde{J}_{n,k}$, the eigenbasis of $(C_{12},C_{34})$ and the one of $(C_{34},C_{234})$, respectively are computed:
\begin{align}
    &\int_{\Delta}dxdy\;x^ay^b(1-x-y)^c \;\widetilde{\Psi}_{p,k}(x,y)\widetilde{J}_{n,k}(x,y)\nonumber  \\
    =&\int_0^1 ds\;s^{a+b+k+1}(1-s)^c\int_0^1dt\,t^a(1-t)^b\hP{a,b}{k}(t)\,  \psi^{(c,d-b-c,2a+b+d+3)}_{p}(s)
    \hP{a,b}{k}(t)\,\hP{a+b+2k+1,c}{n-k}(s)\nonumber\\
    =&N_k^{(a,b)}\int_0^1 ds\;s^{a+b+k+1}(1-s)^c  \psi^{(c,d-b-c,2a+b+d+3)}_{p}(s)
    \,\hP{a+b+2k+1,c}{n-k}(s)\nonumber\\
    =&N_k^{(a,b)}\int_0^1 ds\;s^{a+b+2k+1}(1-s)^c  \psi^{(c,d-b-c,2a+b+d+3+2k)}_{p}(s)
    \,\hP{a+b+2k+1,c}{n-k}(s)\nonumber\\
      =&N_k^{(a,b)}N_{n-k}^{(a+b+2k+1,c)} W_{n-k}^{(a+b+2k+1,c,d-b-c,2a+b+d+3+2k)}(p)\,.
\end{align}
The first equality is computed using the explicit expressions \eqref{eq:Psi5} and \eqref{eq:Jnkt} of $\widetilde{\Psi}_{p,k}(x,y)$ and $\widetilde{J}_{n,k}(x,y)$ and with the change of variable  $s=x+y$ and $t=x/(x+y)$. The second equality is obtained from the orthogonality relation \eqref{eq:ortho} of $\widehat P$. The third is based on the relation 
$\psi^{(c,d-b-c,\delta)}_{p}(s)=s^k\psi^{(c,d-b-c,\delta+2k)}_{p}(s)$. The last one is equivalent to the Koornwinder formula \eqref{eq:KoorF}. Therefore, the functions $\widetilde{\Psi}$ can be expressed in terms of $\widetilde{J}$ with the coefficients given by Wilson polynomials:
\begin{align}
    \widetilde{\Psi}_{p,k}(x,y)=\sum_{n=k}^\infty W_{n-k}^{(a+b+2k+1,c,d-b-c,2a+b+d+3+2k)}(p)\; \widetilde{J}_{n,k}(x,y)\,.
\end{align}

Let us emphasize that this result agrees with the fact that the pair $(C_{12},C_{234})$ generates an $\R_1$ algebra with central data $(C_2,C_{34},C_{1234},C_1)$.

\subsection{Edge \texorpdfstring{$\Phi$--$\widetilde\Psi$}{PPs} (shared generator \texorpdfstring{$C_{12}$}{C12})}\label{sss:PhiPsit}
The eigenfunctions $\Phi_{m,p}$ of $(C_{12},C_{123})$ defined by \eqref{eq:Phi} 
can be transformed using the expansion \eqref{eq:psiP} for $\psi_m^{(b,2p,2a+d+2)}(\frac{x}{x+y})$:
\begin{align}\label{eq:PhiPsit}
    \Phi_{m,p}&= \sum_{k=0}^{\infty}
 W^{(a,b,2p,2a+d+2)}_k\left(m\right) \;\widetilde{\Psi}_{p,k}\,.
\end{align}
This formula provides the overlap coefficients between the eigenbasis $\Phi$ of $(C_{12},C_{123})$ and the one $\widetilde{\Psi}$ of $(C_{12},C_{34})$. This matches again the prediction obtained by looking at the $\R_1$ algebra generated by $C_{123}$ and $C_{34}$.

\subsection{Composite overlap \texorpdfstring{$\Phi$--$\widetilde{J}$}{PtJ} and bivariate Racah polynomials of Tratnik}\label{sss:PhiJ}
Composing the two previous edges that join $\Phi$ to
$\widetilde{J}$ through $\widetilde{\Psi}$ gives
\begin{equation}
\Phi_{m,p}=\sum_{k=0}^{\infty}\sum_{n=k}^{\infty} T_{n,k}(m,p)
\widetilde{J}_{n,k}\,,
\label{eq:PhiJ}
\end{equation}
where 
\begin{equation}
    T_{n,k}(m,p)=  W^{(a,b,2p,2a+d+2)}_k\left(m\right) W_{n-k}^{(a+b+2k+1,c,d-b-c,2a+b+d+3+2k)}(p)\,.
\end{equation}
The overlap coefficients  $  T_{n,k}(m,p)$ is thus the product of two univariate Wilson polynomial: the first one have the parameters depending on the variable of the second one and the parameters of the second one depend on the degree of the first one. This is the bivariate polynomials introduced by Tratnik in \cite{Tratnik1991}.

\section{Conclusion}\label{sec:conc}
We have constructed an explicit embedding of the rank two Racah algebra $\R_2$ into the rank two Jacobi
algebra $\J_2$, by promoting the two multiplication operators of the differential model to Racah type
generators $C_{12},C_{123}$ by using algebraic Heun operatora. This turns the four edges of Figure~\ref{fig:pentJ} corresponding to rank one Jacobi algebras  into rank one Racah algebras as displayed in Figure~\ref{fig:pentR}. We determined the common eigenfunctions of the five commuting
pairs of operators shown in Figure~\ref{fig:pentR} in terms of Jacobi polynomials or Gauss
hypergeometric functions. We computed some overlap coefficients between these eigenbases and, in particular, composing two adjacent edges yields the $\Phi$--$\widetilde{J}$ overlap coefficients as the
bivariate Racah polynomials of Tratnik.

Several directions invite further work. The closed forms obtained here are continuous-spectrum
eigenfunctions. Their truncations provide different finite-dimensional representations of $\R_2$ and the full
recoupling graph (the folded icosidodecahedron described in \cite{crampe2023representations}) deserves a systematic treatment in the present
differential model. The extension to more variables, the $q$-deformation toward the Askey--Wilson level,
and the role of these constructions in the generic superintegrable models \cite{crampe2025super} are
natural avenues.

\appendix
\section{Structure relations of the rank two Jacobi algebra \texorpdfstring{$\J_2$}{J2}}\label{APP A}
The five generators $L,\Lone,\Lthree,\Xone,\Xthree$ satisfy
\begin{equation}
[L,\Lone]=[L,\Lthree]=[\Lone,\Xone]=[\Lthree,\Xthree]=[\Xone,\Xthree]=0,\label{zero}
\end{equation}
together with the relations below.
\subsection{Commutators involving \texorpdfstring{$[L,X_1]$}{LX1}}
\begin{align}
[[L,\Xone],L]&=2\anticomm{\Xone}{L}-2L+2\Lone-(a+b+c+1)((a+b+c+3)\Xone-(a+1)I),\label{LX1L}\\
[[L,\Xone],\Lone]&=0,\\
[[L,\Xone],\Lthree]&=\anticomm{\Xone}{L+\Lthree}+\anticomm{\Xthree-I}{L-\Lone}\nonumber\\
&\quad-(a+b+c+1)\bigl((a+1)(\Xone+\Xthree-I)+(b+1)\Xone\bigr),\\
[[L,\Xone],\Xone]&=-2\Xone^2+2\Xone,\label{LX1X1}\\
[[L,\Xone],\Xthree]&=-2\Xone\Xthree.\label{LX1X3}
\end{align}

\subsection{Commutators involving \texorpdfstring{$[L,X_3]$}{LX3}}
\begin{align}
[[L,\Xthree],L]&=2\anticomm{\Xthree}{L}-2L+2\Lthree-(a+b+c+1)((a+b+c+3)\Xthree-(c+1)I),\label{LX3L}\\
[[L,\Xthree],\Lone]&=\anticomm{\Xone-I}{L-\Lthree}+\anticomm{\Xthree}{L+\Lone}\nonumber\\
&\quad-(a+b+c+1)\bigl((c+1)(\Xone+\Xthree-I)+(b+1)\Xthree\bigr),\\
[[L,\Xthree],\Lthree]&=0,\qquad [[L,\Xthree],\Xone]=[[L,\Xone],\Xthree],\\
[[L,\Xthree],\Xthree]&=-2\Xthree^2+2\Xthree.\label{LX3X3}
\end{align}

\subsection{Commutators involving \texorpdfstring{$[L_1,L_3]$}{L1L3}}
\begin{align}
[[\Lone,\Lthree],L]&=0,\\
[[\Lone,\Lthree],\Lone]&=2\anticomm{\Lone}{\Lthree}+2\Lone^2-2\Lone L+(b+c)(b+1)(L-\Lone-\Lthree)\nonumber\\
&\quad-(b+c)(c+1)\Lthree-(b-c)(a+1)\Lone,\label{L1L3L1}\\
[[\Lone,\Lthree],\Lthree]&=-2\anticomm{\Lone}{\Lthree}-2\Lthree^2+2\Lthree L-(a+b)(b+1)(L-\Lone-\Lthree)\nonumber\\
&\quad+(a+b)(a+1)\Lone+(b-a)(c+1)\Lthree,\label{L1L3L3}\\
[[\Lone,\Lthree],\Xone]&=-\anticomm{\Xone-I}{L-\Lone-\Lthree}-\anticomm{\Xthree}{L-\Lone}\nonumber\\
&\quad+(a+1)\bigl((b+1)\Xthree+(c+1)(\Xone+\Xthree-I)\bigr),\\
[[\Lone,\Lthree],\Xthree]&=\anticomm{\Xone}{L-\Lthree}+\anticomm{\Xthree-I}{L-\Lone-\Lthree}\nonumber\\
&\quad-(c+1)\bigl((a+1)(\Xone+\Xthree-I)+(b+1)\Xone\bigr).
\end{align}

\subsection{Commutators involving \texorpdfstring{$[L_1,X_3]$}{L1X3}}
\begin{align}
[[\Lone,\Xthree],L]&=[[L,\Xthree],\Lone],\\
[[\Lone,\Xthree],\Lone]&=2\anticomm{\Xthree}{\Lone}+\anticomm{\Xone}{\Lone}-2\Lone
-(b+c)\bigl((b+c+2)\Xthree-(c+1)(I-\Xone)\bigr),\label{L1X3L1}\\
[[\Lone,\Xthree],\Lthree]&=[[\Lone,\Lthree],\Xthree],\qquad [[\Lone,\Xthree],\Xone]=0,\\
[[\Lone,\Xthree],\Xthree]&=-2\Xthree^2+2(I-\Xone)\Xthree.\label{L1X3X3}
\end{align}

\subsection{Commutators involving \texorpdfstring{$[L_3,X_1]$}{L3X1}}
\begin{align}
[[\Lthree,\Xone],L]&=[[L,\Xone],\Lthree],\qquad [[\Lthree,\Xone],\Lone]=-[[\Lone,\Lthree],\Xone],\\
[[\Lthree,\Xone],\Lthree]&=\anticomm{2\Xone+\Xthree-I}{\Lthree}
-(a+b)\bigl((a+1)(\Xone+\Xthree-I)+(b+1)\Xone\bigr),\label{L3X1L3}\\
[[\Lthree,\Xone],\Xone]&=-2\Xone(\Xone+\Xthree-I),\qquad [[\Lthree,\Xone],\Xthree]=0.\label{L3X1X1}
\end{align}
Commutators obtainable from the Jacobi identity $[[X,Y],Z]+[[Y,Z],X]+[[Z,X],Y]=0$ are indicated by the
expression they equal rather than repeated.

\subsection{An implied relation}
Using the Jacobi identity and \eqref{LX1L},\eqref{LX1X3} the commutator $[[L,\Xone],[L,\Xthree]]$ can be
computed in two ways, whose comparison yields
\begin{equation}
[\Lone,\Xthree]-[L,\Xthree]=[L,\Xone]-[\Lthree,\Xone].\label{identity}
\end{equation}

\section{Presentation of the rank 2 Racah algebra \texorpdfstring{$\R_2$}{R2} \label{app:algR4}}

The rank two Racah algebra $\R_2$
has been studied in \cite{de2017higher,crampe2021racah,crampe2023representations,post2024racah}.
We recall in this appendix the presentation used in \cite{crampe2023representations}.

The algebra $\R_2$ is generated by the elements
$C_{12}$, $C_{23}$, $C_{34}$, $C_{123}$ and $C_{234}$ as well as central elements $C_{1234}$ and $C_{j}$, $j=1,2,3,4$, subject to different types of defining relations:
\subsection{Commutativity relations.}
\begin{equation}
[C_{12}\,,\,C_{34}]=0\,,\quad [C_{12}\,,\,C_{123}]=0\,,\quad [C_{23}\,,\,C_{123}]=0\,,\quad [C_{23}\,,\,C_{234}]=0\,,\quad
[C_{34}\,,\,C_{234}]=0\,.
\end{equation}

\subsection{\texorpdfstring{$\R_1$}{R1}-type relations.} The five couples $(C_{12},C_{23})$, $(C_{23},C_{34})$, $(C_{123},C_{34})$, $(C_{234},C_{123})$ and $(C_{234},C_{12})$ satisfy relations \eqref{eq:R1} with $(\mathbf{c_1},\mathbf{c_2},\mathbf{c_3},\mathbf{c_{123}})=$ $(C_1,C_2,C_3,C_{123})$,
$\left(C_2\,,C_3\,,C_4\,,C_{234}\right)$,
$(C_{12},C_3,C_4,C_{1234})$,
$(C_4,C_{23},C_1,C_{1234})$ and
$(C_2,C_{34},C_{1234},C_1)$, respectively.
These five $\R_1$-type relations are the ones displayed on Figure~\ref{fig:pentR}.

\subsection{\texorpdfstring{$\R_2$}{R2}-type relations.}
To finish to define $\R_2$, the generators must satisfy the following relations:
\begin{eqnarray}
  [C_{12},C_{23}] + [C_{23},C_{34}] - [C_{123},C_{34}] - [C_{12},C_{234}] + [C_{123},C_{234}] = 0,
\end{eqnarray}
\begin{eqnarray}
\frac12\big[C_{34}\,,\,[C_{12},C_{23}]\big] &=&
C_{12}\big(C_{23}+C_{34}-C_{234}-C_3\big) +C_{23}\big(C_{34}-C_{1234}\big)
-C_{34}\,C_{2}
\nonumber\\
&&+C_{123}\big(C_{234}-C_{34}-C_{2}\big)-C_{234}\,C_3+(C_2+C_3)\,C_{1234}+C_2\,C_3\,,
\quad\\
\frac12 \big[C_{23}\,,\,[C_{123},C_{34}]\big]  &=&
C_{12}\big(-C_{23}+C_{234}-C_4\big) +C_{23}\big(C_{123}-C_{4}\big)
+C_{34}\big(C_{23}-C_{1}\big)
\nonumber\\
&&+C_{123}\big(C_{34}-C_{234}-C_{3}\big)-C_{234}\,C_3+(C_1+C_3)\,C_{4}+C_1\,C_3\,,
\\
\frac12 \big[C_{234}\,,\,[C_{23},C_{12}]\big]  &=&
C_{12}\big(C_{234}-C_4\big) +C_{23}\big(C_{12}-C_{34}+C_{234}-C_{1}\big)
-C_{34}\,C_{1}
\nonumber\\
&&+C_{123}\big(C_{34}-C_{234}-C_{2}\big)-C_{234}\,C_2+(C_1+C_2)\,C_{4}+C_1\,C_2\,,
\\
\frac12 \big[C_{234}\,,\,[C_{123},C_{34}]\big]  &=&
C_{12}\big(-C_{23}+C_{234}+C_4\big)
+C_{34}\big(C_{23} -C_{123} -C_{234} +C_{1234}\big)
\nonumber\\
&&+C_{23}\,C_{1234} +C_{123}\big(-C_{234}+C_{2}\big)+C_{234}\,C_4
\nonumber\\
&&-(C_2+C_{1234})\,C_{4}-C_2\,C_{1234}\,.
\end{eqnarray}
\begin{align}
&\big(C_{12}+C_{234}-C_{2}-C_{1234}\big)\,[C_{12},C_{23}]
+\big(C_{12}-C_2+C_1\big)\,[C_{23},C_{34}]
\nonumber\\
+&\big(C_{123}-C_{23}-C_{12}+C_2\big)\,[C_{12},C_{234}]
+\big(C_{12}-C_2-C_1\big)\,[C_{34},C_{123}]
= 0,\\
&\big(-C_{234}+C_{34}+C_2\big)\,[C_{12},C_{23}]
+\big(C_{12}-C_1+C_2\big)\,[C_{23},C_{34}] \nonumber
\\
+&\big(C_{23}-C_{2}-C_3\big)\,[C_{12},C_{234}]
+2C_2\,[C_{34},C_{123}]
= 0,\\
&\big(-C_{34}-C_3+C_4\big)\,[C_{12},C_{23}] +\big(C_{123}-C_{12}-C_3\big)\,[C_{23},C_{34}] \nonumber
\\
+& 2C_3\,[C_{12},C_{234}] +\big(C_{23}-C_{2}-C_3\big)\,[C_{34},C_{123}]
= 0,\\
&\big(-C_{34}+C_3-C_4\big)\,[C_{12},C_{23}]
+\big(-C_{123}-C_{34}+C_3+C_{1234}\big)\,[C_{23},C_{34}] \nonumber
\\
+&\big(C_{34}-C_{4}-C_3\big)\,[C_{12},C_{234}]
+\big(C_{234}-C_{23}-C_{34}+C_3\big)\,[C_{34},C_{123}]
= 0.
\end{align}

\subsection{\texorpdfstring{$\J_2$}{J2} as a contraction of \texorpdfstring{$\R_2$}{R2}}\label{sec:contraction}

We finally describe how $\J_2$ relates to a contraction of $\R_2$.

Let us rescale the following generators of $\R_2$:
\begin{equation}
C_{12}\to\tau C_{12},\qquad C_{123}\to\tau C_{123},\qquad C_{1234}\to\tau C_{1234},
\end{equation}
and leave unchanged the other generators. Retaining in the relations only the leading terms as
$\tau\to\infty$, a straightforward computation shows that the assignments (up to constant shifts)
\begin{align}
&C_{12}\to\Xthree,\quad C_{23}\to-\Lone+\tfrac14(b+c)(b+c+2),\quad
C_{34}\to-\Lthree+\tfrac14(a+b)(a+b+2),\nonumber\\
&C_{123}\to I-\Xone,\quad C_{234}\to-L+\tfrac14(a+b+c+3)(a+b+c+1),\label{eq:contr}
\end{align}
with central values
\begin{equation}
C_1=0,\quad C_2=\tfrac14(c^2-1),\quad C_3=\tfrac14(b^2-1),\quad C_4=\tfrac14(a^2-1),\quad C_{1234}=1,
\label{eq:contrCent}
\end{equation}
provide a realization of $\J_2$, \textit{i.e.},  recover $\J_2$. Thus the embedding of Section~\ref{sec:embed} and
the contraction \eqref{eq:contr}--\eqref{eq:contrCent} are the two directions relating the rank two Racah
and Jacobi algebras.

\bibliographystyle{unsrt}
\bibliography{ref_R2inJ2}

\end{document}